\documentclass[pra,showpacs,superscriptaddress,twocolumn]{revtex4-1}

\usepackage{amsmath}
\usepackage{amsthm}
\usepackage{amssymb}
\usepackage{amsfonts}
\usepackage{graphicx}
\usepackage{braket}
\usepackage{cases}
\usepackage[mathcal]{eucal}
\usepackage{bbold}

\newcommand{\ante}{ex-ante} 
\newcommand{\antec}{{\ante} control} 
\newcommand{\post}{ex-post} 
\newcommand{\postc}{{\post} control}

\DeclareMathOperator{\Tr}{Tr}
\DeclareMathOperator{\tr}{Tr} 
\DeclareMathOperator{\trhs}{Tr_{HS}}

\DeclareMathOperator{\diag}{diag}

\newcommand{\ketbra}[1]{\ket{#1}\!\bra{#1}} 
\newcommand{\kb}{\ketbra}
\makeatletter

\renewcommand{\le}{\leqslant}
\renewcommand{\ge}{\geqslant}
\newcommand{\C}{\mathbb{C}}
\newcommand{\R}{\mathbb{R}}
\renewcommand{\AA}{\mathcal{A}}
\newcommand{\CC}{\mathcal{C}}
\newcommand{\DD}{\mathcal{D}}
\newcommand{\EE}{\mathcal{E}}
\newcommand{\HH}{\mathcal{H}}
\newcommand{\II}{\mathcal{I}}
\newcommand{\KK}{\mathcal{K}}
\newcommand{\LL}{\mathcal{L}}
\newcommand{\NN}{\mathcal{N}}
\newcommand{\id}{\mathrm{id}}
\renewcommand{\a}{\alpha}
\renewcommand{\b}{\beta}

\newcommand{\ep}{\varepsilon}
\newcommand{\m}{\mu}
\newcommand{\om}{\omega}
\renewcommand{\r}{\rho}
\newcommand{\s}{\sigma}
\newcommand{\ti}{\tilde}

\newcommand{\paren}[1]{{\left( #1 \right)}}
\newcommand{\bigparen}[1]{\bigl({#1}\bigr)}
\newcommand{\Bigparen}[1]{\Bigl({#1}\Bigr)}
\newcommand{\biggparen}[1]{\biggl({#1}\biggr)}
\newcommand{\brac}[1]{{\left\{ #1 \right\}}}

\newcommand{\bigbrac}[1]{\bigl\{{#1}\bigr\}}

\newcommand{\brak}[1]{{\left[ #1 \right]}}
\newcommand{\bigbrak}[1]{\bigl[{#1}\bigr]}
\newcommand{\Bigbrak}[1]{\Bigl[{#1}\Bigr]}

\newcommand{\f}{\frac}
\newcommand{\q}{\quad}
\newcommand{\qq}{\qquad}
\newcommand{\nn}{\notag\\}
\newcommand{\Ad}{\AA}
\newcommand{\inv}{^{-1}}
\newcommand{\norm}[1]{\left\|#1\right\|}
\newcommand{\dr}{``discriminate \& reprepare''}
\newcommand{\dn}{``do nothing''}
\newcommand{\choio}{Choi operator}

\newcommand{\ot}{\otimes}

\newtheorem{theo}{Theorem}
\newtheorem*{theorem*}{Theorem}
\newtheorem*{claim*}{Claim}
\newtheorem{lem}{Lemma}
\newtheorem{conj}{Conjecture}

\begin{document}
\title{Noise suppression by quantum control before and after the noise}
\author{Hiroaki Wakamura}
\email{hwakamura@rk.phys.keio.ac.jp}
\affiliation{Department of Physics, Keio University, Yokohama 223-8522, Japan}
\author{Ry\^uitir\^o Kawakubo}
\email{rkawakub@rk.phys.keio.ac.jp}
\affiliation{Department of Physics, Keio University, Yokohama 223-8522, Japan}
\author{Tatsuhiko Koike}
\email{koike@phys.keio.ac.jp}
\affiliation{Department of Physics, Keio University, Yokohama 223-8522, Japan}
\affiliation{Research and Education Center for Natural Sciences, 
  Keio University, Yokohama 223-8521, Japan}
\date{January, 2017}
\begin{abstract}
We discuss the possibility of protecting the state of a quantum system
that goes through noise,
by measurements/operations before and after
the noise process.
The aim 
is to seek for the optimal protocol 
that
makes the input and output states as close as possible 
and clarify the role of the measurements therein. 
We consider
two cases;
one can perform
quantum measurements/operations
(i) only after the noise process and  
(ii) both before and after that. 
We prove in the two-dimensional Hilbert space that, 
in the case (i), the noise suppression is  
essentially impossible for all types of noise and, 
in the case (ii), 
the optimal protocol for the depolarizing noise is 
either the {\dn} protocol or 
the {\dr} protocol.
These protocols are not ``truly quantum'' and 
can be considered as classical.
They involve no measurement or
only use the measurement outcomes.
These results 
describe the fundamental limitations in 
quantum mechanics
from the viewpoint of control theory. 
Finally, we conjecture that a statement similar to 
the case (ii) holds for higher-dimensional Hilbert
spaces and present some numerical evidence. 

\end{abstract}
\maketitle

\section{Introduction} \label{sec:intro}

Measurement in quantum theory substantially 
differs from that in classical theory.
One cannot identify the state of a system by 
measurement on a single sample.
In addition, measurement always disturbs the system.
The limitation of manipulating quantum systems
can be understood as being imposed by these 
characteristics of quantum measurement.
An example is the no-cloning theorem~\cite{nocloning} 
which states that it is impossible to create identical 
copies of an arbitrary quantum state. 
If one could make a clone, 
then one could extract complete information 
from a single state by creating infinitely many copies 
thereof, which contradicts quantum mechanics.
Other examples are the facts that 
one cannot discriminate non-orthogonal states 
perfectly
and that 
one cannot measure 
non-commuting observables without errors. 
In presence of such impossibility, 
many researchers study 
how well one can perform these tasks
mentioned above. 
Imperfect cloning~\cite{buzhil,werner98}, 
state discrimination~\cite{helstrom,unambig123} 
and 
uncertainty relations for noise 
and disturbance~\cite{ozawaan,watanabe} 
are examples of 
the studies that make 
a quantitative assessment of
ability to realize the tasks approximately.

We would like to discuss 
an aspect of the limitations of quantum operation 
that one cannot protect states against 
noise. 
Here we use the word ``noise'' in a wide sense so that
it refers to any irreversible dynamics 
induced by environments. 
In classical systems, one can protect a state against the
irreversible dynamics (noise)  if accurate measurements and operations
can be done 
and if the state is not a statistical mixture, 
by taking the complete record of the state before the noise affects the 
system. 
In quantum theory, it is not the case even if the state is 
pure.
Measurement cannot be done accurately and 
disturbs the state.
Nevertheless,
one can still consider operations which
approximately reverse the noise. 
Such approximate operations 
reveal the limit beyond which the noise cannot be 
suppressed any further.
It also is interesting to understand the role which 
measurement plays for the task.

Noise suppression, the attempt to protect a certain class of
states against given noise, is an important 
problem in the field of quantum
control.  Many researchers are working on the problem in various
ways. 
Some try to protect a few states, while others try to
protect all the pure states. 
The approaches are further classified by 
whether one uses \textit{{\post}}\/ control only or \textit{{\ante}}\/ and \textit{{\post}}\/ control together, 
where we mean by {\ante} and {\postc}
the quantum measurements/operations performed 
before and after, respectively, the noise process. 
For the problem of protecting two states, Bra\'{n}czyk {\it et
  al}.~\cite{bramen07} obtained the optimal {\postc} for the dephasing
noise. 
After a while, Mendon\c{c}a {\it et al}.~\cite{mengil08}  gave a
method for constructing the optimal {\postc} for arbitrary noise. 
Compared to the protection of two states, 
protecting all the pure states is more
difficult and challenging. 
Zhang {\it et al.}~\cite{zhang08} pointed out a part of the
difficulty; they prove that {\postc} alone cannot suppress the
depolarizing noise at all. 
Korotkov and Keane~\cite{korotkov04} considered {\antec} as well as
{\postc} and found that it is possible to protect, 
to some extent, all the pure states of a qubit
against the amplitude damping noise. 
Along this line, Wang {\it et al.}~\cite{wang14} made a further study
using numerical methods. 
We remark that there are still many other approaches 
to reducing the effect of decoherence
in a wider context~\cite{Sho95,Kni96,Reim05,Lid98,VioLlo98}.

\begin{figure}[t]
 \centering 
\includegraphics[width=8.5cm]{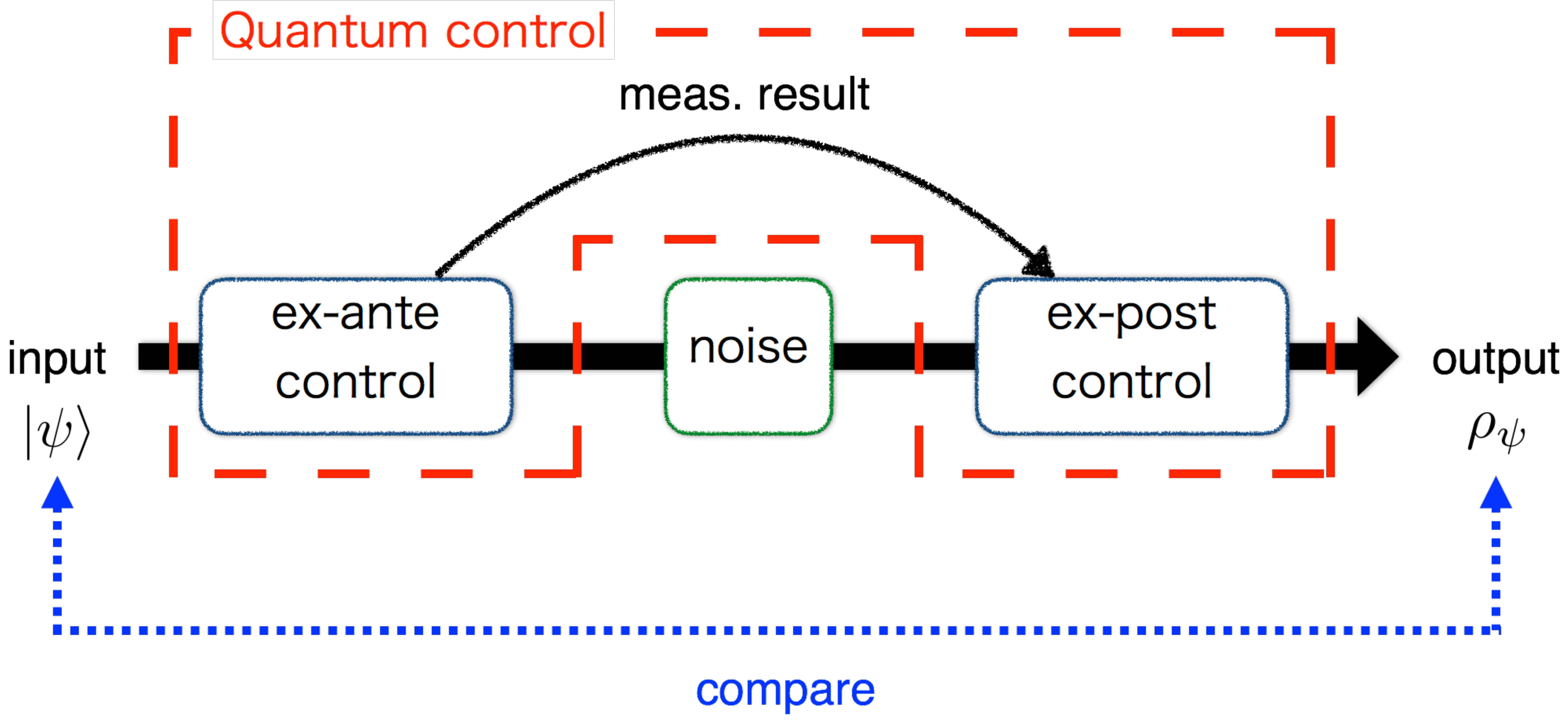}
\caption{The schematic diagram of the quantum control (\antec{} and \postc{}).}
\label{fig:antepost}
\end{figure}

In this paper, we discuss protection of all the pure states against noise
by \ante{} and \post{} control (see Fig.~\ref{fig:antepost}).  
First, we restrict ourselves to 
the scheme with 
\postc{} alone and 
discuss if it can suppress the effect of noise. 
In the case of the two-dimensional Hilbert space (a qubit), 
we prove that the suppression is impossible 
for any type of noise (Theorem~\ref{theo:kekka1}). 
This result, generalizing a part of Ref.~\cite{zhang08}, reveals the decisive
necessity of {\antec} 
if one wishes to protect all the pure states.
Next, we consider the scheme with both 
\ante{} and \post{} control.
We focus on the depolarizing noise, 
whose isotropic property make the suppression task more difficult,
and thus suitable to see the  improvement thanks to {\antec}.
Then we show, in Theorem~\ref{theo:kekka2}, that the optimal control protocols are rather simple and easy to understand classically and {\antec} is useful only when the noise is strong (Theorem~\ref{theo:kekka2}). 
Finally, we conjecture that essentially the same results hold in the case of higher dimensional Hilbert spaces, which is supported by a numerical calculation in the case of three-dimensional Hilbert space.

The paper is organized as follows.
In Sec.~\ref{sec:pre}, the mathematical tools used in 
our discussion 
are
reviewed. 
In Sec.~\ref{sec:qucon}, the setup of quantum control with 
\ante{} and \post{} control is explained. 
In Sec.~\ref{sec:poscon}, noise suppression for a qubit 
by \postc{} alone
is discussed in general. 
We consider 
in Sec.~\ref{sec:anposcon} 
\ante{}-\postc{} 
for a qubit 
under the depolarizing noise 
and present a result on the optimal protocols, 
which is proved in Sec.~\ref{pfkekka2}. 
In Sec.~\ref{sec:numcon}, 
we propose a conjecture for higher-dimensional Hilbert spaces 
and show some numerical evidence. 
Sec.~\ref{sec:sum} is devoted to conclusion and discussions. 
The Appendix lists the theorems used in the text.

\section{Basics of quantum operations}
\label{sec:pre}

In this section, we shall introduce basic mathematical tools and
notation used in 
our analysis. 
Throughout the paper, we consider 
physical
systems which are
represented by 
a 
finite-dimensional 
Hilbert space. 
Let $\mathcal{H}$ be  such a Hilbert space. 
Let $\mathcal{L(H)}$ be the set of all linear
operators on $\mathcal{H}$.  

An operator $O\in\mathcal{L(H)}$ is said positive and denoted by 
$O\ge0$ if 
$\braket{\psi|O|\psi}\ge0$ holds for any $\ket{\psi}\in \mathcal{H}$. 
A quantum state is described by a density operator $\rho\in\LL(\HH)$ 
such that $\rho\ge0$ and $\tr \rho=1$. 
{\it Fidelity}\/ measures closeness of two states 
and is 
defined by  
\begin{align}
 F(\rho,\ketbra{\psi}):=\braket{\psi|\rho|\psi} 
 \label{eq-def-fidelity}
\end{align}
in the case that at least one of the two states is 
pure~\cite{Nielsen}.
The fidelity represents the probability of measuring $\ket{\psi}$ when
the state is 
$\rho$, and it is equal to unity if and only if $\rho=\ketbra{\psi}$.

A linear map $\mathcal{E}:\mathcal{L(H)}\rightarrow\mathcal{L(H)}$ is
said  {\it positive}\/ if $O\ge0$ implies
$\mathcal{E}(O)\ge0$. 
The map $\mathcal{E}$
is said {\it completely positive (CP)}\/ if the map 
$\mathcal{E}\otimes {\id}_n$ 
is positive for every positive integer $n$, 
where $\id_n$ denotes the identity map on $\mathcal{L}(\C^n)=\C^{n\times n}$. 
The map $\mathcal{E}$ is
said 
{\it trace-preserving}\/ if $\tr\mathcal{E}(O)=\tr O$  for
any $O\in\mathcal{L(H)}$. 
A trace-preserving completely positive map is called a {\it TPCP map}. 
It is known that
any physical evolution of a quantum state corresponds to a TPCP map, 
and vice
versa~\cite{Nielsen}. 

An example of CP map is 
\begin{align}
\Ad_A(O):=AOA^\dag
\label{eq-sand}
\end{align}
defined for $A\in\mathcal{L(H)}$. 
It is known that
any CP map can be expressed as a sum of such $\AA_A$ 
(the Kraus representation). 
The CP map $\AA_A$ is trace-preserving only if $A$ is a
unitary operator. 
Another example of TPCP map, which is important 
in our discussion below, is 
mixing with the completely mixed state, or 
{\it
  depolarizing noise} 
$\mathcal{D}_\ep$, 
defined by 
\begin{align}
 \mathcal{D}_\varepsilon(\rho):=(1-\varepsilon)\rho+\varepsilon\frac{1}{d}\tr
 \rho, 
 \label{depo-noise}
\end{align}
where 
$\varepsilon$ is a parameter between 0 and 1,  
$d$ is the dimensionality of $\mathcal{H}$ 
and the operator $1/d$ is
the completely mixed state. 
The map 
outputs the completely mixed state with probability $\varepsilon$ and 
leaves the input state untouched with probability $1-\varepsilon$. 

A family 
$\{\II_\omega\}_{\omega\in\Omega}$ of CP maps 
with $\sum_{\omega\in\Omega}\II_\omega$ being trace-preserving is
called a {\em CP instrument}.  
It is known that any physical measuring process corresponds to a
CP instrument, and vice versa~\cite{Ozawa84}. 
In this paper, we assume that the number of the measurement outcomes is
finite. 
The state evolution by the measurement is described as 
\begin{align}
 \rho \mapsto \frac{\II_\omega(\rho)}{\tr\II_\omega(\rho)},\quad
\text{with probability} \quad
\tr\II_\omega(\rho).
\end{align}

A set 
$\brac{M_\om}{}_{\om\in\Omega}$ 
of positive operators on $\HH$
such that 
$\sum_{\om\in\Omega} M_\om=1$ is called a {\em POVM}. 
A CP instrument $\brac{\Tr\II_\om}_{\om\in\Omega}$ defines a POVM by 
$\Tr\II_\om(\rho)=\Tr\rho M_\om$. 
We shall say that 
such 
a POVM $\brac{M_\om}{}_{\om\in\Omega}$ 
and 
a CP instrument 
$\brac{\II_\om}{}_{\om\in\Omega}$ 
are associated with each other. 
A POVM has the information on all statistical properties of the 
measurement outcomes, while a CP instrument $\brac{\II_\om}$ 
has still more information on the state after the measurement. 
Any CP instrument associated with $\brac{M_\om}{}_{\om\in\Omega}$ 
can be written as (See Hayashi~\cite{Hayashi06}, p.189, Theorem~7.2)
\begin{align}
  \II_\om=\KK_\om\circ \Ad_{\sqrt{M_\om}}, 
\label{lem:hayashi}
\end{align}
where 
$\mathcal{K}_\om$ 
is 
a TPCP map. 
We will call $\brac{\AA_{\sqrt{M_\om}}}_{\om\in\Omega}$ 
a simple CP instrument associated with 
$\brac{M_\om}_{\om\in\Omega}$. 

The set $\LL(\HH)$ can be 
regarded 
as a Hilbert space with 
the Hilbert-Schmidt inner product
$\braket{X,Y}_{\mathrm{HS}}:=\Tr X^\dagger Y$. 
Then a linear map $\EE$ on $\LL(\HH)$ is a linear operator 
on the Hilbert space $\mathcal{L(H)}$. 
Thus the trace of $\mathcal{E}\in\LL\paren{\LL(\HH)}$  
is defined as 
\begin{align}
 \trhs\mathcal{E}:=\sum_{i} \braket{V_i,\mathcal{E}(V_i)}_{\mathrm{HS}},
\end{align}
where 
$\{V_i\}_i$ is an orthonormal basis of 
the Hilbert space $\mathcal{L(H)}$.
For example, when $\dim\HH=2$, 
the set of Pauli operators $\{\sigma_\mu/\sqrt2\}_{\mu=0}^3$, 
with $\sigma_0$ being the identity operator, 
is an orthonormal basis of
$\LL(\HH)$ so that 
the trace of $\EE\in\LL\paren{\LL(\HH)}$ can be written as 
\begin{align}
 \trhs\mathcal{E}
 =
 \f12
 \sum_{\mu=0}^3
 \tr\brak{\sigma_\mu\mathcal{E}(\sigma_\mu)}. 
 \label{eq-trhs-2d}
\end{align}

\section{The setup} \label{sec:qucon}

In this section, 
we shall present the main problem and its mathematical
formulation. 

We consider the {\em \ante{}-\post{} quantum control scheme} 
defined by the following sequence of processes 
(depicted in Fig.~\ref{fig:antepost}): 
\begin{enumerate}
 \item ``state preparation''
\\
An unknown state $\ket{\psi}\in\mathcal{H}$ is prepared. 
 \item ``\antec{}''
\\
A measurement is performed, 
which is described by a CP instrument
$\{\II_\om\}_{\om\in\{1,\dots, M\}}$, where $M$ is a positive
integer. 
 \item ``noise''
\\
The state undergoes an undesired evolution, called  ``noise,''
described by a TPCP map $\mathcal{N}$. 
 \item ``\postc{}''
\\
An operation, 
which depends on the
measurement outcome $\om$ of the \ante{} control, 
is performed on the system. 
This is described by a family of TPCP maps,
$\{\CC_\om\}_{\om\in\{1,\dots,M\}}$, which we call the \postc{}. 
\end{enumerate}
For given noise $\NN$, an 
\ante{}-\post{} control protocol is specified by 
the family $\{(\II_\om,\CC_\om)\}_{\om\in\{1,\dots,M\}}$. 
We assume throughout the paper that
the prepared state is pure, 
though one can consider more general mixed 
state preparation. 
We also assume that the state $\ket\psi$ above 
is 
{\em completely unknown}, i.e., 
the probability distribution is uniform on the unit sphere in
$\mathcal{H}$. 

The problem that we want to consider is, 
for given noise $\NN$, 
to find an {\em optimal}\/ scheme 
such that 
the states after the measurement with outcome $\om$ 
are as similar to the original state $\ketbra\psi$ as
possible.  
For defining the optimality, it is natural to introduce some
evaluation function $h$
and take an average with respect to the probability of obtaining $\om$, 
and then take an average with respect to $\ket\psi$ which is
completely unknown,

\begin{widetext}
\begin{align}
  &
  \int_{\|\ket{\psi}\|=1} d\psi
  \sum_{\substack{\text{measurement}\\ \text{outcomes} } } 
  \text{evaluation function}
  (\text{final state, initial state})
  \times
  \text{ probability}
  \nn
  =&
  \int_{\|\ket{\psi}\|=1} d\psi
  \sum_\om
  h\paren{\CC_\om\circ \NN \paren{\f{\II_\om(\kb\psi)}{\Tr\II_\om(\kb\psi)} }, 
    \ketbra{\psi}} 
  \Tr\II_\om(\kb{\psi}), 
\end{align}
\end{widetext}
where 
the integral is over all
unit vectors in $\mathcal{H}$ with the uniform measure normalized by
$\int_{\norm{ \ket{\psi} }=1} d\psi=1$. 
We choose the function $h$ to be 
the fidelity $F$ in Eq.~\eqref{eq-def-fidelity}. 
An advantage of the choice is 
that 
the resulting total evaluation function, the {\it average fidelity} 
\begin{align}
 \bar{F}=\int_{\|\ket{\psi}\|=1} d\psi
 \braket{\psi|\EE(\ketbra\psi)|\psi},
 \label{eq:avefid1} 
\end{align}
depends on the protocol $\brac{(\II_\om,\CC_\om)}_\om$ only through 
the {\em average operation}\/ 
\begin{align}
\EE:=
 \sum_{\om=1}^M\CC_\om\circ\mathcal{N}\circ\II_\om 
 \label{eq-averaged-op}
\end{align}
which is a TPCP map. 

We close the section by presenting a 
useful formula for the average fidelity. 
\begin{lem}
The average fidelity  $\bar F$ is given by 
\begin{align}
 \bar{F}=
 \frac{1}{d(d+1)}
 \left(d+\trhs\EE\right), 
\label{eq:avefid2} 
\end{align}
where 
$\EE$ is 
the average operation 
\eqref{eq-averaged-op} 
of the protocol $\brac{(\II_\om,\CC_\om)}_\om$ and 
$d:=\dim \mathcal{H}$. 
\label{lem:useful}
\end{lem}

\begin{proof}
Let $S$ be the 
swap operator, 
$S(\ket{\psi}\ket{\phi}):=\ket{\phi}\ket{\psi}$, 
or 
in an orthonormal basis $\brac{\ket{ij}}_{i,j}$, 
$S=\sum_{ij}\ket{ij}\bra{ji}$. 
One can easily see that 
\begin{align}
  \tr\left[AB\right]=\tr\left[S(A\otimes B) \right]. 
  \label{eq-trabs}
\end{align}
It follows from 
\eqref{eq:avefid1} 
and 
\eqref{eq-trabs} 
that 
\begin{align}
  \bar F
=&
\tr\left[ 
  S(\id \otimes \mathcal{E}) (Q)
   \right], 
 \label{eq-barf-intermediate}
\end{align}
where 
$
Q:=\int_{\|\ket{\psi}\|=1} d\psi
\ketbra{\psi} \otimes
\ketbra{\psi}
$. 
Then the result is easily obtained by the following formulas, 
\begin{align}
  &Q=\f{1+S}{d(d+1)}, 
  \label{eq-symmetrizer}
  \\
  &\tr\left[S(\id \otimes \mathcal{E})(S)\right]
  =
  \trhs \EE. 
  \label{eq-trhse}
\end{align}
Indeed, from 
\eqref{eq-trabs}, 
\eqref{eq-barf-intermediate}, 
\eqref{eq-symmetrizer} and 
\eqref{eq-trhse}, one has 
\begin{align}
\paren{d(d+1)}\bar F&=
\tr\left[S(\id\otimes \mathcal{E})(1) \right]
+\tr\left[S(\id \otimes \mathcal{E})(S)\right]
\nn
&
=\tr\EE(1)+\trhs\mathcal{E}
=d+\trhs\mathcal{E}. 
\end{align}
Let us show \eqref{eq-symmetrizer} and 
\eqref{eq-trhse}. 
Eq.~\eqref{eq-symmetrizer} is seen~\cite{popescu05} by 
noting that 
$Q$ 
commutes with $U\otimes U$ for any $U$ in $\mathrm{SU}(d)$, 
the special unitary group on $\C^d$.
By Schur's lemma, 
$Q$ acts 
as scalar operators on the 
symmetric and
antisymmetric subspaces of 
$\HH\otimes\HH$, 
which are the spaces of irreducible representations. 
The scalar factors are found easily. 
Eq.~\eqref{eq-trhse} is seen by direct calculation, 
LHS = $
\sum_{ijkl} \bra{ij}
\paren{\ket k\!\bra l\otimes\EE
  \paren{\ket l\!\bra k}}\ket{ji}
=
\sum_{ij} \bra{j}
\EE
\paren{\ket j\!\bra i}
\ket{i}
$ = RHS. 
This completes the proof.  
\end{proof}

\section{\postc{}} \label{sec:poscon}

In this section, we shall consider the noise suppression by 
\postc{} only (Fig.~\ref{fig:post}) and present 
our first main result. 

The control sequence has no branches and the protocol is determined by 
a TPCP map $\mathcal{C}$ describing the \post{} control. 
Thus our aim is to find the optimal \postc{} $\CC$. 
The average fidelity  to be maximized is 
\begin{align}
 \bar{F}=
 \frac{1}{d(d+1)}
\left(d+\trhs\CC\circ\NN\right),
\end{align}
from Lemma~\ref{lem:useful}.

\begin{figure}[t]
 \centering
 \includegraphics[width=8.5cm]{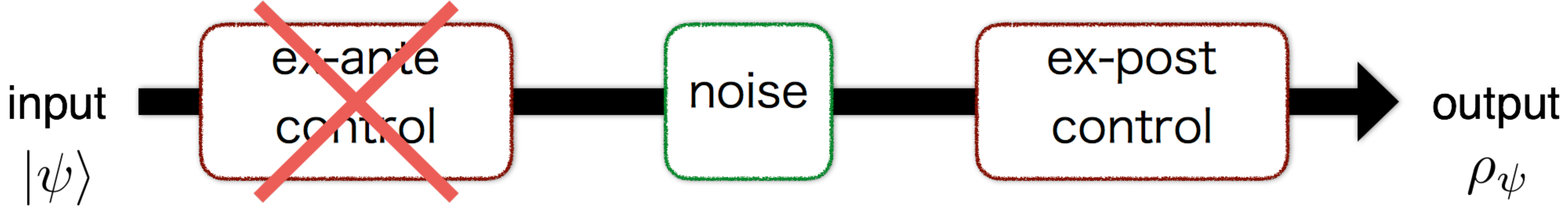}
\caption{The schematic diagram of the quantum control (only \postc{}).}
\label{fig:post}
\end{figure}

\begin{theo}
 For any noise $\NN$ in the two-dimensional Hilbert space $\HH$,
 the optimal \post{} control protocol 
 is a unitary transformation. 
\label{theo:kekka1}
\end{theo}

\begin{proof}
In the two-dimensional Hilbert space $\HH$, 
we can express a general TPCP map $\EE$ in the basis
$\brac{\s_\mu}_{0\le \mu\le3}$ such that 
\begin{align}
{\mathcal{E}}(1)&=1+\sum_{k=1}^3 t_k\sigma_k, \q
{\mathcal{E}}(\sigma_j)=\sum_{j,k=1}^3s_{jk}\sigma_k. 
\end{align}
A necessary condition for positivity and trace-preserving property is 
$\EE(1\pm\s_j)\ge0$, or 
\begin{align}
  \sum_{k=1}^3|t_k\pm s_{jk}|^2\le1.   
  \label{eq-pos}
\end{align}
In the two-dimensional Hilbert space, 
one can transform the noise $\NN$ to $\ti\NN$ 
by unitary operations for the input and 
output states 
such that $\ti \NN$ has 
diagonal $s_{jk}$, namely, one has 
\begin{align}
 \mathcal{N}=\Ad_U\circ\, \tilde{\mathcal{N}}\circ\Ad_V 
\end{align}
with some unitary operators 
$U$ and $V$ on $\mathcal{H}$, 
where $\AA$ is defined in \eqref{eq-sand}. 
Introducing $\tilde\CC:=\Ad_{V}\circ\,\mathcal{C}\circ\Ad_{U}$, 
one has 
\begin{align}
 \trhs\mathcal{C}\circ\mathcal{N}
&=\trhs\tilde\CC\circ\tilde{\mathcal{N}} 
=
\f12
\sum_{\m=0}^3
\tr\sigma_\m\tilde\CC\circ\tilde{\mathcal{N}}(\sigma_\m)
\notag\\ 
&=1+\frac{1}{4}\sum_{j=1}^3
\tr{\sigma_j\tilde\NN(\sigma_j)}\, 
\tr{\sigma_j\tilde\CC(\sigma_j)},
\end{align}
where we have used \eqref{eq-trhs-2d} and trace
preserving property of $\tilde\CC$. 
Because $\tilde\CC$ is TPCP, \eqref{eq-pos} implies 
$\tr\sigma_j\tilde\CC(\sigma_j)\le2$ for $j=1,2,3$, so that 
\begin{align}
  \trhs\mathcal{C}\circ\mathcal{N} \le 1
  +
  \f12
  \sum_{j=1}^3
  \tr{\sigma_j\tilde\NN(\sigma_j)}. 
\end{align}
The equality holds only when 
$\tr\sigma_j\tilde\CC(\sigma_j)=2$, so that, 
by \eqref{eq-pos} again, 
one has $\tilde\CC=\id$. 
Thus one has 
 $\mathcal{C}=\Ad_{V}^{-1}\circ\Ad_{U}^{-1}$,
 which is a unitary transformation.
Then the maximum average fidelity $\bar{F}$ is 
$(6+\sum_j\tr{\sigma_j\tilde\NN(\sigma_j)})/12$. 
\end{proof}

Intuitively, this result says that essentially 
no \postc{} can suppress the effect of noise
if we are completely ignorant of 
the initial state.
On the other hand, when we have some knowledge of the initial state then we can suppress the effect of noise by quantum control~\cite{bramen07,mengil08}.
Instead of restricting the candidates for initial state, 
we consider \antec{} to extract some information
of the initial state, 
which is discussed in the next section.

From the proof, we find that 
the quantity 
$(1/2)\sum_j\tr{\sigma_j\tilde\NN(\sigma_j)}$ 
characterizes to what extent the noise 
is reversible
and allows a geometrical interpretation. 
In the Bloch sphere representation of the state, 
the whole sphere is mapped to an ellipsoid. 
Each $(1/2)\tr{\sigma_j\tilde\NN(\sigma_j)}$ is the contraction rate along
a principal axis.

In closing this section, we remark that 
an \ante{}-\post{} control with a single branch 
is equivalent to 
an \post{} control. 
In fact, the average fidelity can be written as 
\begin{align}
 \bar{F}=\frac{1}{d(d+1)}
 \Bigl(d+\trhs\II\circ\mathcal{C}\circ\mathcal{N}\Bigr) 
\end{align}
by Lemma~\ref{lem:useful} and the cyclic property of the trace.
Here, $\II,\mathcal{C}$ are TPCP maps corresponding to \antec{} and
\postc{}, respectively, but the same fidelity can be achieved 
by an \postc{} $\II\circ\mathcal{C}$.

\section{\ante{}-\postc{}} \label{sec:anposcon}

In this section, we discuss the noise suppression by \ante{}-\postc{} 
(Fig.~\ref{fig:antepost}) on a qubit. 
We shall consider the case of two branches 
because a POVM consisting of two projections
gives the optimal discrimination between 
completely unknown states of a qubit~\cite{MasPop95}.

In general, there is a trade-off between 
the information gained and 
the disturbance caused by the ex-ante control.
The \antec{} can extract some information on the initial state
which may be useful for noise suppression. 
At the same time, the measurement disturbs the state.
Thus one might expect that a protocol 
with a ``soft'' \ante{} measurement  could be optimal. 
It turns out, however, that this is not the case.

\begin{theo}
 Let the noise $\NN$ be 
 the depolarizing noise 
 $\mathcal{D}_\varepsilon$ 
 defined in
 \eqref{depo-noise}. 
 For the two-dimensional Hilbert space $\HH$, 
 the optimal \ante{}-\postc{} protocol $\brac{(\II_\om,\CC_\om)}_{\om=1,2}$ 
 is given as follows.

   (i) When the noise is weak, 
   $\ep\le2/3$, 
   the {\dn} protocol
   is optimal, 
   which is given by
   \begin{align}
     \II_\om\propto \id,\quad \CC_\om=\id, \quad\om=1,2. 
   \end{align}
   The optimal average fidelity is 
   $\bar F_{\mathrm{DN}}=1-\ep/2$. 

   (ii) When the noise is strong, 
   $\ep\ge2/3$, 
   the {\dr} protocol is optimal, 
   which is given by
   \begin{align}
     \II_\om(\rho)&=\ketbra{\phi_\om}\rho\ketbra{\phi_\om},
     \\
     \mathcal{C}_\om(\rho)&=\ket{\phi_\om}\bra{\phi_\om}{\tr\r},
   \end{align}
   where $\brac{\ket{\phi_\om}}_{\om=1,2}$ 
   is an arbitrary orthonormal basis of $\mathcal{H}$. 
   The optimal average fidelity is  
   $\bar F_{\mathrm{DR}}=2/3$. 
 \label{theo:kekka2}
\end{theo}
We give the proof in the next section. 

The {\dn}  protocol literally does nothing, and merely lets the
system undergo the noise. 
The value of the average fidelity 
$\bar F_{DN}$ 
in Theorem~\ref{theo:kekka2}
can be obtained by direct calculation. 
Namely, one substitutes $\EE=\DD_\ep$ into the definition
\eqref{eq:avefid1} of $\bar F$ and has 
\begin{align}
 \bar F_{\textrm{DN}} = \int_{\|\ket{\psi}\|=1}
\biggparen{
(1-\ep) \braket{\psi|\psi}\!\braket{\psi|\psi}
+ \ep\f 1 2
 \braket{\psi|\psi}
}
=1-\frac{\ep}{2}.
\label{eq-attain-dn}
\end{align}

The {\dr} protocol 
means that one measures and discriminates between a certain two 
orthogonal states 
before the noise process and reprepares the 
corresponding state after the noise process. 
The value of the average fidelity 
$\bar F_{\mathrm{DR}}$
in Theorem~\ref{theo:kekka2}
can be calculated as follows.
Without loss of generality,  
one can choose $\ket{\phi_\om}=\ket{\om}$, 
$\om=0,1$, 
where 
$\s_3\ket{0}=\ket0$ and 
$\s_3\ket{1}=-\ket1$. 
Substituting $\braket{\om|\s_3|\om}=(-1)^\om$ and
$\braket{\om|\s_1|\om}=\braket{\om|\s_2|\om}=0$ 
into \eqref{eq-trhs-2d},
one obtains $\trhs\EE=2$ and hence 
\begin{align}
 \bar F_{\mathrm{DR}}=\f 1 6 \biggparen{2+\trhs\EE}=\f 2 3
\label{eq-attain-dr}
\end{align}
from Lemma~\ref{lem:useful}. 
It is known from the study of imperfect cloning~\cite{MasPop95,chiribella10}
that any \dr{} protocol
with arbitrary number of branches 
does not give a larger value.

If the noise is infinitesimally weak, it is natural that 
doing nothing is better than the other protocols. 
If the noise is so strong that the state is completely destroyed after
the noise process, 
one should obtain information on the initial
state as much as possible before the system goes through the noise process. 
The result implies that there is no intermediate regime 
where the
optimal protocol involves weak measurements. 
The {\dr}  protocol only uses the classical information
extracted by the \ante{} measurement while the {\dn}  protocol perform no
quantum measurement or operation. 
These protocols are rather classically motivated and are not ``truly
quantum,'' in that they do not reflect any trade-off relation between
information gain and disturbance. 
It is remarkable that these classical control protocols are better
than any other quantum control protocols. 
The result shows that we cannot suppress the noise even if we can
perform \antec{}. 
This may be understood as fundamental limitations in quantum mechanics.

\section{Proof of Theorem~\ref{theo:kekka2}}
\label{pfkekka2}
In this section, we give a proof of Theorem~\ref{theo:kekka2}, 
after showing two lemmas. 
\begin{lem}
 Let $\HH$ be two-dimensional. 
 For any TPCP map $\mathcal{E}$, there exists a TPCP map
 $\tilde{\mathcal{E}}$ that satisfies 
\begin{align}
 \mathcal{D}_\ep\circ\mathcal{E}=\tilde{\mathcal{E}}\circ\mathcal{D}_\ep.
\end{align}
\label{lem:koukan}
\end{lem}
\begin{proof}
  The cases $\ep=0,1$ are trivial so that we assume $\ep\ne0,1$. 
  Then there exists $\DD_\ep\inv$ 
  which is actually 
  $\DD_{-\ep/(1-\ep)}$. 
  One therefore has
  \begin{align}
    \tilde{\EE}=\DD_\varepsilon\circ\EE\circ
    \DD_{-\ep/(1-\ep)}. 
  \end{align}
  We prove that the map 
  $\ti\EE$ is CP, 
  though 
  $\DD_{-\ep/(1-\ep)}$ is not. 
  Any Hermiticity-preserving linear map 
  $\sum_{\m=0}^3 x^\m\sigma_\m
  \mapsto\sum_{\m=0}^3 x'{}^\m\sigma_\m$
  is specified by a linear map 
  from $\R^4$ to $\R^4$, 
  $(x^\m)\mapsto (x'{}^\m)$, $\m=0,1,2,3$. 
  From a theorem by Ruskai {\it et al.} (see the Appendix), the TPCP map
  $\EE$ is
  expressed as 
  \begin{align}
   \begin{pmatrix}
      x'{}^0 \\
      \mathbf{x}'
    \end{pmatrix}
    =
    \begin{pmatrix}
      1 & 0 \\
      \mathbf{t}
      &
      E
    \end{pmatrix}
    \begin{pmatrix}
      x^0 \\
      \mathbf{x}
    \end{pmatrix}, 
\end{align}
where 
$\mathbf{x}=(x^1,x^2,x^3)^T$, 
$\mathbf{t}\in\mathbb{R}^3$, $E$ is a $3\times3$ matrix, 
and $\mathbf{t}$ and the (signed) singular values $d_i$ of $E$ satisfy 
\eqref{eq-pos-abs}, 
\eqref{eq:cptp1}, \eqref{eq:cptp2} and \eqref{eq:cptp3}. 
In the same way, $\ti\EE$ is expressed as 
\begin{align}
   \begin{pmatrix}
      x'{}^0 \\
      \mathbf{x}'
    \end{pmatrix}
    =
    \begin{pmatrix}
      1 & 0 \\
      (1-\ep)\mathbf{t}
      &
      E
    \end{pmatrix}
    \begin{pmatrix}
      x^0 \\
      \mathbf{x}
    \end{pmatrix}.
\end{align}
Since $0<\ep<1$, 
the components of the matrix above also satisfy the conditions 
\eqref{eq-pos-abs}, 
\eqref{eq:cptp1}, \eqref{eq:cptp2} and \eqref{eq:cptp3}. 
Thus $\ti\EE$ is TPCP. 
\end{proof}

\begin{lem}
  Let $A$, $\ti R$, $D$ and $R$ be real square matrices. 
  If $A$ and $D$ are diagonal and 
  $\tilde R$
  is orthogonal, the following inequality holds, 
  \begin{align}
    \Tr\brak{A R D\tilde R}\le 
    \sum_j |d_j| \biggparen{ \sum_i a_i^2r_{ij}^2 }^{1/2}, 
  \end{align}
  where 
  $a_j$ and $d_j$ are the diagonal elements of $A$ and
  $D$, respectively, and 
  $r_{ij}$ is the $(i,j)$ element of $R$. 
  \label{lem-schwa}
\end{lem}
\begin{proof}
Let $\ti r_{ij}$ be the $(i,j)$ 
element of $\tilde R$. 
One has 
\begin{align}
  \Tr\brak{ARD\ti R}
  =\sum_{ij}a_i r_{ij}d_j\ti r_{ji}
  =\sum_{j}|d_j|\sum_i a_i r_{ij}
  \frac{d_j\tilde r_{ji}}{|d_j|}. 
\end{align}
One can view the $i$-sum as an inner product of vectors 
$(a_ir_{ij})_i$ and 
$(d_j\ti r_{ji}/|d_j|)_i$. 
Applying the Cauchy-Schwarz inequality to the inner product, 
and using the orthogonality of $\ti R$, 
one is able to show the claim. 
\end{proof}

\begin{proof}[Proof of Theorem \ref{theo:kekka2}]
The proof consists of five parts. 
The main idea is to show that $\bar F$ does not exceed the values
attained by the {\dn} and {\dr} protocols. 

{\em Step 1. Reduction to simple CP instrument.}
By Lemma~\ref{lem:useful}, 
the optimal protocol 
$\{(\II_\om,\CC_\om)\}_{\om=1,2}$ 
is the maximizer of 
$\sum_{\om=1}^2f_\om$, with 
\begin{align}
f_\om:=\trhs\CC_\om\circ\mathcal{D}_\varepsilon\circ\II_\om. 
\label{eq-fom}
\end{align}
As was explained in \eqref{lem:hayashi}, 
the CP instrument 
$\brac{\II_\om}$ 
is specified by 
a family $\brac{\mathcal{K}_\om}$ of TPCP maps and a POVM $\brac{M_\om}$ 
so that 
$\II_\om
=
\mathcal{K}_\om\circ \Ad_{\sqrt{M_\om}}$. 
By Lemma~\ref{lem:koukan}, 
there is a TPCP map $\tilde{\mathcal{K}}_\om$ which satisfies
$\mathcal{D}_\varepsilon\circ\mathcal{K}_\om
=\tilde{\mathcal{K}}_\om\circ\mathcal{D}_\varepsilon$. 
From these, $f_\om$ can be written as 
\begin{align} 
f_\om 
&=
\trhs\CC_\om
\circ
\mathcal{D}_\varepsilon
\circ
\mathcal{K}_\om 
\circ 
\Ad_{\sqrt{M_\om}}
\nn
&=
\trhs
\CC_\om
\circ
\tilde{\mathcal{K}}_\om 
\circ
\mathcal{D}_\varepsilon
\circ
\Ad_{\sqrt{M_\om}}
\nn
&=
\trhs
\tilde\CC_\om 
\circ\mathcal{D}_\varepsilon\circ\Ad_{\sqrt{M_\om}}, 
\label{eq-fom-abstract}
\end{align}
where 
$\tilde\CC_\om:=\CC_\om\circ\tilde{\mathcal{K}}_\om$. 
Thus, 
the original optimization problem for 
the protocol 
$\brac{(\CC_\om,\II_\om)}$ 
is
translated to that for a protocol 
$\bigbrac{ 
  ( 
  \tilde\CC_\om,
  \Ad_{\sqrt{M_\om}}
  )
}$ 
specified by 
a family 
of TPCP maps 
$\bigbrac{\tilde{\mathcal{C}}_\om}$ 
and a POVM 
$\brac{M_\om}$. 

Let us choose a basis $\brac{\ket0,\ket1}$ 
in which 
$M_1$ and $M_2=1-M_1$ are diagonal, so that 
one has 
\begin{align}
  &M_\om=\a_\om\ketbra{0}+\b_\om\ketbra{1}, 
  \\
  &\a_\om,\b_\om\ge0, 
  \q
  \sum_\om \a_\om=
  \sum_\om \b_\om=1. 
\end{align}
From \eqref{eq-trhs-2d} and \eqref{eq-fom}, one has 
$f_\om
=
\f12\paren{
  \tr M_\om
  +
  I_\om
}$ 
so that 
\begin{align}
  f_1+f_2
  =
  1+\frac{1}{2}(I_1+I_2), 
\end{align}
where 
\begin{align}
&I_\om
:=
\sum_{j=1}^3\tr \brak{
  \s_j
  \tilde\CC_\om 
  \circ\mathcal{D}_\varepsilon\circ\Ad_{\sqrt{M_\om}}(\s_j)
  }
  \nn
  =&\,
  (1-\ep)
  \sqrt{\a_\om\b_\om}\, 
  \sum_{j=1}^2\tr \brak{
  \s_j
  \tilde\CC_\om (\s_j)
  }
  \nn
  +
  &\, 
  \f{\a_\om-\b_\om}2
  \tr\brak{
    \s_3
    \tilde\CC_\om (1)
  }
  +
  (1-\ep)
  \f{\a_\om+\b_\om}2
  \tr\brak{
    \s_3
    \tilde\CC_\om (\s_3)
  }. 
\end{align}

{\em Step 2. 
Necessity for each
$\tilde{\mathcal{C}}_\om$ to be extreme.}
From here to the end of Step 4, 
we fix the POVM $\brac{M_\om}$ and vary the TPCP maps 
${\tilde \CC_\om}$ to obtain a bound for each $I_\om$. 
At this stage one can treat each $I_\om$ independently. 
We drop the subscript $\om$ from the variables and parameters till 
Step 4. 

Since 
$f$ is a linear functional of $\ti\CC$, 
we observe that 
the optimal $\tilde{\mathcal{C}}$ must be one of the extreme 
points in the convex space of TPCP maps. 
From a theorem by Ruskai {\it et al.}\/ (see the Appendix), 
such a TPCP map can be written in the form 
\eqref{eq-map-matrix} with the condition 
\eqref{eq-extreme-TPCP}. 
Thus one has 
\begin{align}
  I&= 
  (\a-\b)
  d_0
  r_{33}
  +(1-\ep)
  \tr\left[
    A
    R
    D
    \ti R
  \right], 
  \label{eq-Iom-2}
\end{align}
where $R$ and $\tilde R$ are real $3\times3$ rotation matrices, 
$r_{33}$ is the $(3,3)$ element of $R$, 
$A:=\diag[
2\sqrt{\a\b}, 
2\sqrt{\a\b}, 
\a+\b]$, 
$D:=
\diag[d_1, 
d_2, 
d_1d_2]
$, 
$d_0^2:=(1-d_1^2)(1-d_2^2)$ and $d_1, d_2\in[-1,1]$. 

{\em Step 3. An $\ti R$-independent upper bound of $I$.} 
We shall derive an upper bound $J$ of $I$, which is independent
from $\ti R$. 
From Lemma~\ref{lem-schwa}, one has 
\begin{align}
  \Tr\brak{A R D\tilde R}\le 
  \sum_j |d_j| \biggparen{ \sum_i a_i^2 r_{ij}^2}^{1/2}, 
  \label{eq-schwa}
\end{align}
where 
$a_i$ and $d_i$ are diagonal elements of $A$ and
$D$, respectively, and 
$r_{ij}$ is the $(i,j)$ component of $R$. 

Next,  it follows 
from 
orthogonality of $R$ 
that 
the matrix with $(i,j)$ element being 
$r_{ij}^2$
is doubly stochastic. From 
the Birkhoff-von Neumann theorem (see the Appendix), 
such a matrix must be a convex combination of permutation matrices. 
Furthermore, because we are considering the case 
$a_1=a_2$ $(=2\sqrt{\a\b})$ 
in the maximization of $I$, 
it is enough to consider the case 
\begin{align}
  &
  (r_{ij}^2)_{0\le i,j\le3}
  =
  \nn
  &
  (1-2p)
  {\small
    \begin{pmatrix}
    1&&\\
    &1&\\
    &&1
  \end{pmatrix}
  }
  +
  (p+q)
  {\small\begin{pmatrix}
    &&1\\
    &1&\\
    1&&
  \end{pmatrix}
}
  +
  (p-q)
  {\small
    \begin{pmatrix}
    1&&\\
    &&1\\
    &1&
  \end{pmatrix}
}, 
  \label{eq-perm}
\end{align}
where $0\le q\le p\le 1/2$. 
From 
\eqref{eq-Iom-2}, 
\eqref{eq-schwa} and \eqref{eq-perm}, 
one has 
\begin{align}
&I\le J
:= 
\left|(\a-\b)d_0\right|
\sqrt{1-2p}\notag\\
&+(1-\ep)\biggl[
|d_1|\sqrt{a_1^2+(a_3^2-a_1^2)(p+q)}
\notag\\
&
+|d_2| \sqrt{a_1^2+(a_3^2-a_1^2)(p-q)}
+|d_1d_2| \sqrt{a_3^2-2(a_3^2-a_1^2)p}
\biggr]. 
\label{eq-Itilde}
\end{align}
The bound $J$ is a function of $(p,q,d_1,d_2)$ while $a_1$ and
$a_3$ are parameters.

{\em Step 4. Joint concavity and a $\ti \CC$-independent bound.} 
Let 
$x:=(d_1^2+d_2^2)/2$, $y:=(d_1^2-d_2^2)/2$. 
Then $J(p,q,x,y)$ 
is 
jointly concave with respect to variables $q$ and $y$, 
which can be easily seen by direct calculation of 
the Hessian of each 
term of $J$. 
Furthermore, 
$J$ is 
invariant 
under 
$(q,y)\mapsto (-q,-y)$. 
Thus one has 
\begin{align}
 &J(p,q,x,y)
 =
 \frac{1}{2}\left(J(p,q,x,y)+J(p,-q,x,-y)\right)\\
&\le J(p,0,x,0)
\nn
&=
|\a-\b|(1-x) 
\sqrt{1-2p}\notag\\
&\qq
+(1-\ep)\Bigl[
2\sqrt x\sqrt{4\a\b+(\a-\b)^2p}
\nn
&
\qq \qq \qq \qq
+x \sqrt{(\a+\b)^2-2(\a-\b)^2p}
\Bigr],
\label{eq-bound-of-J}
\end{align}
where we have substituted the expressions of $a_1$, $a_2$ and $d_0$. 
Because $\a,\b\in[0,1]$ and $p\in[0,1/2]$, 
one can show 
by a simple observation 
that the bound \eqref{eq-bound-of-J} does not exceed 
\begin{align}
(1-x) 
\sqrt{r}
+2(1-\ep)
\sqrt x\sqrt{(\a+\b)^2-r}
+(1-\ep)(\a+\b)x, 
\label{eq-boundtmp}
\end{align}
where $r:=(\a-\b)^2(1-p)\in[0,1]$. 
Applying the Cauchy-Schwarz inequality to 
the first two terms of 
\eqref{eq-boundtmp}, 
one sees that 
\eqref{eq-boundtmp} does not exceed 
\begin{align}
  (\a+\b)\paren{
    \sqrt{(1-x)^2+4(1-\ep)^2 x}+(1-\ep) x
  }. 
\end{align}
This is a convex function of $x$ 
hence reaches the maximum at the
boundary $x=0,1$. 
Thus, from \eqref{eq-Itilde}, 
we arrive at a $\ti\CC$-invariant upper bound, 
\begin{align}
  I\le (\a+\b)\max\brac{3(1-\ep),1}. 
  \label{eq-ineq-J}
\end{align}

{\em Step 5. A protocol-independent bound for $\bar F$ and its attainability.} 
We revive the subscript $\om$. 
Because $\sum_\om \a_\om=\sum_\om \b_\om=1$, 
one immediately obtains 
a 
protocol-independent 
upper bound 
from 
\eqref{eq-ineq-J}, 
$I_1+I_2\le 2\max\brac{3(1-\ep),1}$. This is equivalent to 
\begin{align}
  \bar F\le \f16\paren{3+\max\brac{3(1-\ep),1}}. 
\end{align}
The values 
$\bar F=1-\ep/2$ 
and 
$\bar F=2/3$ 
are attained by {\dn} and {\dr} protocols, respectively, 
as was shown in 
\eqref{eq-attain-dn} and 
\eqref{eq-attain-dr}. 
\end{proof}

\section{Higher dimensions} \label{sec:numcon}

\begin{figure}[t]
 \centering
 \includegraphics[width=8.5cm]{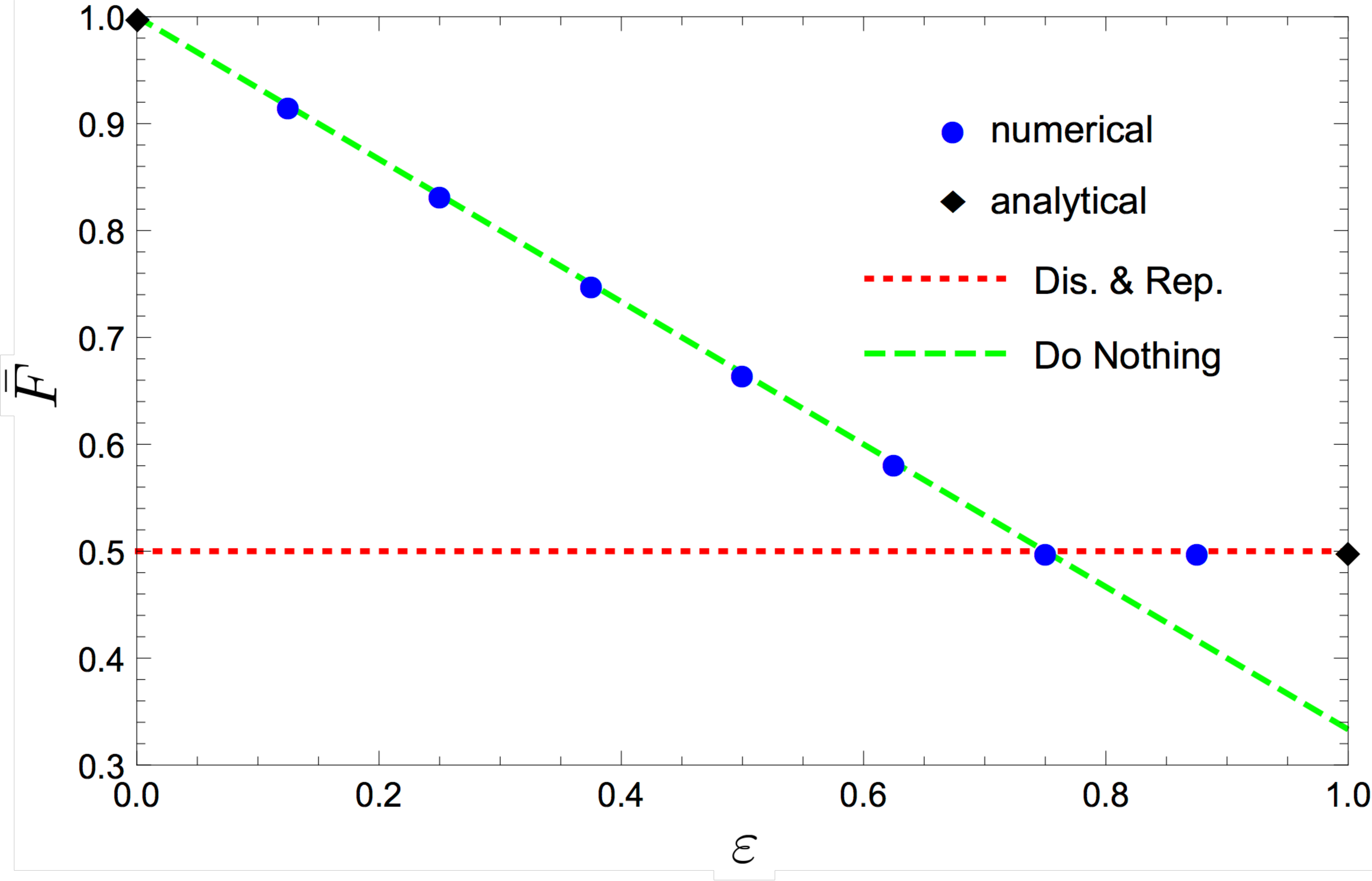}
 \caption{(Color online) 
   The maximum average fidelity $\bar F$ achieved by 
   \ante-\postc{} for the depolarizing noise 
   of strength $\ep$, when $\dim\HH=3$. 
   The (blue) round points are the maxima obtained by numerical
   optimization, 
  the (black) diamond points are exact values for 
  $\ep=0,1$, 
  the (green) dashed line 
  is 
  $\bar F$ achieved by the {\dn} protocol, 
  and 
  the (red) dotted line is that obtained by the {\dr} 
  protocol. 
The results support Conjecture~\ref{conj:kekka3}.} 
\label{fig:numres}
\end{figure}

Though Theorem~\ref{theo:kekka2} is valid only for 
the two-dimensional Hilbert space $\HH$, 
similar results may hold in higher dimensions.
In this section, we propose a conjecture in general dimensionality and 
show a numerical evidence for the three-dimensional Hilbert space. 

\begin{conj}
 Let $\mathcal{D}_\varepsilon$ be the depolarizing noise 
 \eqref{depo-noise} in the $d$-dimensional Hilbert space $\HH$. 
 Then the optimal {\ante}-{\postc} 
 $\{(\II_\om,\CC_\om)\}_{\om=1}^d$
 for $\DD_\ep$ is given as follows.

(i) When the noise is weak, $\ep\le d/(d+1)$, 
   the {\dn} protocol is optimal, which is given by
\begin{align}
  \II_\om\propto \id,\quad \CC_\om=\id, \quad\om=1,\dots,d. 
\end{align}
The optimal average fidelity is 
$\bar F_{DN}=1-(d-1)\ep/d$. 

(ii) When the noise is strong, $\ep\ge d/(d+1)$, 
the {\dr}  protocol is optimal, which is given by
\begin{align}
 \II_\om(\rho)=&\ketbra{\phi_\om}\rho\ketbra{\phi_\om},\\ 
 \CC_\om(\rho)=&\ketbra{\phi_\om}\tr\r, \quad\quad\quad\quad\quad 
 \om=1,\dots,d,
\end{align}
where $\{\ket{\phi_\om}\}_{\om=1}^d$ is an arbitrary orthonormal basis of $\mathcal{H}$.
The optimal average fidelity is 
$\bar F_{DR}=2/(d+1)$. 
\label{conj:kekka3}
\end{conj}
The proof of 
Theorem~\ref{theo:kekka2} does not work 
in the same way
because we used a concrete 
characterization of the extreme points of TPCP
maps when $\dim\HH=2$.

In our numerical calculation below,
we make use of Choi's correspondence (see the Appendix) 
which relates the CP maps  $\CC_\om, \II_\om\in\LL(\LL(\HH))$
with positive operators 
$C_\om, I_\om \in \LL(\HH\ot\HH)$, 
called the Choi operators, 
as in \eqref{eq:choi}.

\begin{lem}
The average fidelity for an \ante-\postc{} protocol 
$\brac{(\II_\om,\CC_\om)}_{\om=1}^d$ 
is given by 
\begin{align}
 \bar{F}=&
 \frac{1}{d(d+1)}
 \Bigparen{d+\sum_{\om=1}^d
 \tr\brak{R_\om (\id\otimes\mathcal{N})(S) }}, 
\end{align}
where 
\begin{align}
 R_\om:=&\tr_C\left[(1_A\otimes C_{\om,CB}^T) (I_{\om,AC}\otimes
 1_B) \right], 
\quad \om=1,\dots,d, 
\label{eq-num}
\end{align}
is the partial transpose of the
{\choio} of $\II_\om\circ\CC_\om$ 
[\,$\HH_A$, $\HH_B$ and $\HH_C$ are copies of $\HH$, 
$\Tr_C$ denotes partial trace on $\HH_C$, and $X_{AC}$ denotes 
an operator $X\in\LL(\HH_A\ot\HH_C)$, etc.]. 
\label{lem-fid-in-choi}
\end{lem}
\begin{proof}
  From Lemma~\ref{lem:useful}, it suffices to show 
  \begin{align}
    \trhs \CC_\om\circ \NN\circ \II_\om
    =
     \tr\brak{R_\om(\id\otimes\mathcal{N})(S) }. 
   \end{align}
   For any linear maps $\EE_1$ and $\EE_2$ on $\LL(\HH)$ 
   such that $\EE_1$ and $\EE_2^*$ commute, 
   one has 
   \begin{align}
     (\EE_1\ot \EE_2)^*(S)=(\EE_2\ot\EE_1)(S), 
     \label{eq-action-on-swap}
   \end{align}
   where an asterisk denotes the dual linear map.
   Eq.~\eqref{eq-action-on-swap} is seen by 
   \eqref{eq-trabs} because 
   both of 
   $\Tr \brak{ (A\ot B)(\EE_1\ot \EE_2)^*(S) }$
   and 
   $\Tr \brak{ (A\ot B)(\EE_2\ot \EE_1)(S) }$ 
   are equal to $\Tr\brak{ \EE_1(A)\EE_2(B) }$. 
   One has
   \begin{align}
     &\trhs \CC_\om\circ \NN\circ \II_\om
     =
     \trhs \II_\om\circ\CC_\om\circ \NN
     \nn
     &=
     \tr\brak{ S (\id\ot(\II_\om\circ\CC_\om\circ \NN))(S) }
     \nn
     &=
    \tr\brak{ R_\om (\id\ot \NN))(S) }, 
  \end{align}
  where $R_\om=((\II_\om\circ\CC_\om)\ot\id)(S)$ 
  is the partial transpose on $\HH_B$ of the {\choio} of $\II_\om\circ\CC_\om$
  and have used the cyclic property of the trace, 
  \eqref{eq-trhse} and 
  \eqref{eq-action-on-swap}. 
  Furthermore, one has 
  \begin{align}
    R_\om
    &=
    (\II_\om\ot\id)\paren{\brak{(\CC_\om\ot\id)(S^{T_B})}{}^{T_B}}
    \nn
    &
    =
    (\II_\om\ot\id)(C_{\om}{}^{T_B})
    \nn
    &=
    \Tr_C\brak{ (1_A\ot C_{\om,CB}{}^{T_CT_B})(I_{\om,AC}\ot\id_B) }, 
  \end{align}
  where $T_B, T_C$  denotes the partial transpose on $\HH_B, \HH_C$ and 
  we have used the fact that 
  $S^{T_B}$ is an unnormalized maximally entangled state
  $\kb\Psi$, $\ket\Psi=\sum_i\ket{ii}$
  and \eqref{eq:choi}. 
\end{proof}

By Lemma~\ref{lem-fid-in-choi}, 
the problem of finding optimal $\bar F$ and $\brac{(\II_\om,\CC_\om)}$ 
is recast in the following form: 
\begin{align}
&\text{maximize}\quad f:=\sum_{\om=1}^d
\tr\brak{R_\om (\id\otimes\DD_\ep)(S)} \notag\\
&\text{subject to} \quad C_\om, I_\om \ge0,\; P=0, 
\label{eq-maximization-problem}
\end{align}
where $R_\om$ is defined in \eqref{eq-num} and 
\begin{align}
 P:=\sum_\om \tr \Bigbrak{ \bigparen{\tr_AC_\om-1}^2 }
+\tr \Bigbrak{ \bigparen{\tr_A\sum_\om I_\om-1}^2 } 
\end{align}
is the ``penalty function.'' 
The conditions $C_\om, I_\om\ge0$ imply 
that complete positivity of $\CC_\om, \II_\om$ and 
the condition $P=0$ ensures the trace-preserving property of 
$\CC_\om$ and $\sum_\om\II_\om$ 
(see the Appendix). 
When $\ep=0,1$, the problem can be exactly solved for general
$d=\dim\HH$. 
For $\ep=0$, the {\dn} protocol obviously 
attains the maximum $\bar F=1$.
For $\ep=1$, it is easy to see that 
all the protocols fall into {\dr} protocols 
defined by POVMs.
The maximal average fidelity 
achieved by such protocols 
is $\bar F=2/(d+1)$~\cite{chiribella10}. 

One can solve the maximization problem~\eqref{eq-maximization-problem} 
in the following steps. 
\begin{enumerate}
\item 
Generate $2d$ lower triangular matrices 
$\brac{L_{C_\om},L_{I_\om}}_{1\le\om\le d}$ 
so that its nontrivial components are random numbers which obey 
uniform distribution in the interval 
$\bigbrak{-\sqrt{d},\sqrt{d}}$. 
          The last is a necessary condition for $P=0$.
\item 
Set the Choi operators as 
$C_\om=L_{C_\om}L_{C_\om}^\dag$ and $I_\om=L_{I_\om}L_{I_\om}^\dag$. 
Then $C_\om, I_\om \ge0$ automatically hold. 
\item 
Apply a numerical maximization method to 
$f-\lambda P$, where $\lambda$ is a (large) positive number. 
The penalty term $-\lambda P$ effectively ensures the  
condition $P=0$. 
\end{enumerate}

We examined Conjecture~\ref{conj:kekka3} when $\dim\HH=3$. 
We carried out the numerical scheme above 
for 5000 initial random points, 
with $\lambda=10^3$ and 
the maximization method being the simulated annealing. 
By randomness in the initial data and in the
optimization scheme, 
we expect that the global maximum of $\bar F$ are found. 
Fig.~\ref{fig:numres} 
shows the optimal average fidelity 
$\bar F$ as a function of the noise strength $\ep$. 
The results suggest 
that 
either 
the {\dn} or {\dr} protocol 
is optimal, 
depending on the strength of the noise. 
This provides evidence for Conjecture~\ref{conj:kekka3}.

\section{Conclusion and discussions} \label{sec:sum}

We have discussed the problem of 
protecting completely unknown
states against given noise by \ante{} and \postc{}
scheme.  
A protocol in the scheme is 
described mathematically by 
$\brac{(\II_\om,\CC_\om)}_{\om\in\Omega}$ 
where $\brac{ \II_\om }$ is 
the CP instrument with the set $\Omega$ of outcomes 
applied before the system goes through
the noise and 
$\CC_\om$ are 
the TPCP maps applied after 
the system suffered the noise. 
To evaluate the closeness of the input and output states, 
we have chosen the average fidelity $\bar F$ between the input and output states, 
which is linear in $\II_\om$ and
$\CC_\om$. 
We have shown in 
Theorem \ref{theo:kekka1}
that when the scheme involves \postc{} only, 
one  essentially cannot suppress any given noise. 
In other words, 
all 
one can do is to cancel the unitary rotational 
part of the noise, 
if one is completely ignorant of the input state. 
Next, 
we have considered \ante{}-\postc{} scheme, 
focusing on the depolarizing noise.
We have shown in 
Theorem~\ref{theo:kekka2} 
that the optimal average fidelity is achieved by 
protocols which are 
not ``truly quantum,'' or 
can be understood classically. 
Namely, if the noise is weak, 
the {\dn} protocol is optimal, which does nothing to the system
literally, and no other protocol can make a larger average fidelity. 
If the noise is strong, 
the {\dr} protocol is optimal.
In the protocol,
one completely measures the system
beforehand, discards the resulting state  
and reconstructs 
the state estimated from the measurement.
The theorems above are for the two-dimensional Hilbert space. 
Finally, we 
have proposed Conjecture~\ref{conj:kekka3} that Theorem~\ref{theo:kekka2} is
essentially true for any dimensionality, or more precisely, 
the optimal average fidelity is achieved either by the {\dn} or {\dr}
protocol. 
We have found numerical evidence to support the conjecture 
in the three-dimensional Hilbert space.

A natural question to ask is whether the result similar to
Theorem~\ref{theo:kekka2} holds or not for noise other than the
depolarizing noise.
This is not the case
at least for the amplitude damping noise,
e.g. spontaneous emission;
it is shown numerically that 
there exists a quantum protocol that can do better than the {\dn} and
{\dr} protocols~\cite{wang14}. 
Our result suggests that 
noise suppression for completely unknown input states 
is impossible 
except by using the bias or anisotropy of the noise itself 
even when one includes the {\antec} in the scheme. 
We note that the depolarizing noise is isotropic in the sense that 
it commutes with arbitrary unitary operations. 
Thus our result may be understood as describing 
the fundamental limitations of quantum 
mechanics from the viewpoint of noise suppression. 
It is worth pursuing which class of noise 
allows nontrivial suppression 
and identifying the optimal controls therein. 
In particular, 
it may be important to examine whether Theorem~\ref{theo:kekka2} 
can be extended 
to all unital noise, the noise that preserves the completely mixed state. 
Our numerical calculations 
suggest 
that it is true at least for the dephasing noise.


Let us discuss our results further 
in the fundamental aspect: irreversibility of quantum processes. 
While unitary operations describe reversible processes only, TPCP maps
include irreversible processes. 
Then 
it is natural to ask 
to what extent a
given TPCP map has irreversibility.
For an operator on a Hilbert space, 
the polar decomposition extracts its
``irreversible'' part uniquely. 
However, as far as we know, there is no such a simple and canonical
decomposition for TPCP maps, which are operators on Banach spaces.
This fact makes it difficult to define the irreversible part of a
given TPCP map.
In some sense, our work is an attempt to address this problem from the
viewpoint of control theory.
Using \ante{} and \post{}, we try to cancel an effect of a given TPCP
map (noise) and define operationally the irreversible part 
as what still remains.
In Theorem \ref{theo:kekka1}, we sought approximate left inverse of a given TPCP map and
find it to be a unitary operation. 
It is not trivial that the approximate left inverse is
unique and reversible, which is the conclusion of our theorem.
In Theorem \ref{theo:kekka2}, using both \ante{} and \postc{}, 
we found that
the irreversible part of the depolarizing noise is itself 
when the noise is weak.

\section*{Acknowledgments}

T.K acknowledges the support from MEXT-Supported Program for the Strategic Research Foundation at Private Universities “Topological Science” and from Keio University Creativity Initiative “Quantum Community.”

\appendix*
\section{Theorems used in the text}

In this Appendix, we shall quote some mathematical facts used in the
main text. 

A square matrix is said doubly stochastic if 
all components are nonnegative and if the sum over 
any row is unity and the sum over any column is unity. 
\begin{theorem*}[Birkoff-von Neumann~\cite{BvN}]
  Any doubly stochastic matrix 
  is 
  a convex combination
  of permutation matrices. 
\end{theorem*}
There is one-to-one correspondence between CP maps and positive
operators on a larger Hilbert space. 
\begin{theorem*}[Choi~\cite{Choi75}]
Let $\mathcal{H}_A$ be a $d$-dimensional Hilbert space 
and let $\mathcal{H}_B$ be a copy of $\mathcal{H}_A$. 
Let $\{\ket{i}\}_{i\in\{1,2,\dots,d\}}$ be an orthonormal basis for 
each of $\mathcal{H}_A$ and $\mathcal{H}_B$. 
Then there is a one-to-one correspondence between 
a CP map
$\EE\in\LL\paren{\LL(\HH_A)}$ 
and a positive operator $E\in\LL\paren{\HH_A\ot\HH_B}$ such that 
\begin{align}
 \mathcal{E}(\rho)=\tr_B\left[(1\ot\rho^T)E\right],
\end{align}
where 
$T$ denotes the transpose with respect to the 
basis above. 
The operator $E$ is called the 
{\it \choio}\/ for $\mathcal{E}$ and is explicitly written as 
\begin{align}
 E:=(\mathcal{E}\ot \id)
 \paren{
   \kb\Psi
   }
\in \LL(\mathcal{H}_A\otimes\mathcal{H}_B), 
\label{eq:choi}
\end{align}
where $\ket\Psi:=\sum_i\ket{ii}$ 
is an unnormalized maximally entangled state. 
The CP map $\EE$ is trace-preserving if and only if 
$\Tr_A E=1$. 
\end{theorem*}

Let $\HH$ be two-dimensional and $\EE$ be a Hermiticity-preserving 
linear map on $\LL(\HH)$. 
By a parametrization
\begin{align}
  &\mathcal{E}\biggparen{
    \sum_{\m=0}^3x^\m\sigma_\m
  }
  =
  \sum_{\m=0}^3x'{}^\m\sigma_\m, 
  \label{eq-matrix-expression}
\end{align} 
$\EE$ is expressed as a linear map $\R^4\to\R^4$, 
$(x^\m)\mapsto (x'{}^\m)$, ${\m=0,1,2,3}$. 
If $\EE$ is positive and trace-preserving, 
then there exist real $3\times3$ rotation matrices $R, \tilde{R}$ 
and 
real numbers 
$d_i,t_i\;(i=1,2,3)$ such that 
\begin{align}
  &|t_i|+|d_i| \le 1, 
  \label{eq-pos-abs}
  \\
  \begin{pmatrix}
    x'^0 \\
    x'^1 \\
    x'^2 \\
    x'^3 \\
  \end{pmatrix}
  &=
  \begin{pmatrix}
    1 & 0 \\
    R
    \begin{pmatrix}
      t_1\\t_2\\t_3
    \end{pmatrix}
    &
    R
    \begin{pmatrix}
      d_1&&\\&d_2&\\&&d_3
    \end{pmatrix}
    \tilde R
  \end{pmatrix}
  \begin{pmatrix}
    x^0 \\
    x^1 \\
    x^2 \\
    x^3 
  \end{pmatrix}. 
  \label{eq-map-matrix}
\end{align}
Ruskai {\em et al.}\/ gave 
the concrete parametrization of TPCP maps when $\HH$ is
two-dimensional, 
extending the work by 
Fujiwara and Algoet~\cite{Fuji99}. 

\begin{theorem*}[Ruskai-Szarek-Werner
~\cite{BethRuskai2002159}, Corollary 2 and Theorem 4] 
  Let $\EE$ 
  a positive, trace-preserving linear map 
  specified by 
  \eqref{eq-matrix-expression}, 
  \eqref{eq-pos-abs} and 
  \eqref{eq-map-matrix}. 

  (i) 
  The map 
  $\EE$ 
  is completely positive if and only if all of the following
  inequalities hold: 
  \begin{align} 
    &(d_1+d_2)^2\le (1+d_3)^2-t_3^2-(t_1^2+t_2^2)\left(\frac{1+d_3\pm
        t_3}{1-d_3\pm t_3}\right), 
    \label{eq:cptp1} \\ 
    &(d_1-d_2)^2\le (1-d_3)^2-t_3^2-(t_1^2+t_2^2)\left(\frac{1-d_3\pm
        t_3}{1+d_3\pm t_3}\right),
    \label{eq:cptp2}\\ 
    &\left[1-(d_1^2+d_2^2+d_3^2)-(t_1^2+t_2^2+t_3^2)\right]^2
    \notag\\ 
    &\ge
    4\left[d_1^2(t_1^2+d_2^2)+d_2(t_2^2+d_3^2)+d_3^2(t_3^2+d_1^2)-2d_1d_2d_3\right]. 
    \label{eq:cptp3} 
  \end{align}

 (ii) The map $\mathcal{E}$ is 
  in the closure of the set of extreme points of
  the space of TPCP maps 
  if and only if 
  there exist $R$, $\tilde R$ and $d_1, d_2\in[-1,1]$ such that 
  \begin{align}
    d_3=d_1d_2, 
    \q t_1=t_2=0, 
    \q t_3^2=(1-d_1^2)(1-d_2^2). 
    \label{eq-extreme-TPCP} 
  \end{align}
  \label{lem:ruskai}
\end{theorem*}

\end{document}